\documentclass{IEEEtran4PSCC}
\usepackage{multirow}
\usepackage{amsmath}
\usepackage{array}

\usepackage{cite}[compressed/ranged]
\usepackage[colorlinks=true, citecolor=blue]{hyperref}

\ifCLASSINFOpdf
   \usepackage[pdftex]{graphicx}
\else
   \usepackage[dvips]{graphicx}
\fi
\usepackage{comment}
\usepackage{amsmath}
\usepackage{amsfonts,amssymb}
\usepackage{amsthm}
\usepackage{booktabs}

\theoremstyle{definition} %
\newtheorem{definition}{Definition}

\theoremstyle{plain} %
\newtheorem{theorem}{Theorem}

\theoremstyle{plain} %
\newtheorem{lemma}{Lemma}

\theoremstyle{plain} %
\newtheorem{proposition}{Proposition}

\usepackage{comment}
\usepackage{romannum}
\usepackage{tabularx}
\makeatletter
\let\old@ps@headings\ps@headings
\let\old@ps@IEEEtitlepagestyle\ps@IEEEtitlepagestyle
\def\psccfooter#1{%
    \def\ps@headings{%
        \old@ps@headings%
        \def\@oddfoot{\strut\hfill#1\hfill\strut}%
        \def\@evenfoot{\strut\hfill#1\hfill\strut}%
    }%
    \def\ps@IEEEtitlepagestyle{%
        \old@ps@IEEEtitlepagestyle%
        \def\@oddfoot{\strut\hfill#1\hfill\strut}%
        \def\@evenfoot{\strut\hfill#1\hfill\strut}%
    }%
    \ps@headings%
}
\makeatother

\psccfooter{%
        \parbox{\textwidth}{\hrulefill \\ \small{24th Power Systems Computation Conference} \hfill \begin{minipage}{0.2\textwidth}\centering \vspace*{4pt} \includegraphics[scale=0.06]{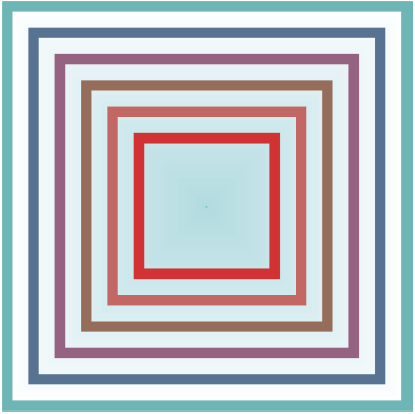}\\\small{PSCC 2026} \end{minipage} \hfill \small{Limassol, Cyprus --- 08-12 June, 2026}}%
}

\begin{document}
%

\title{Pooling Probabilistic Forecasts for Cooperative Wind Power Offering}


\author{\IEEEauthorblockN{Honglin Wen\IEEEauthorrefmark{1}\IEEEauthorrefmark{2},
Pierre Pinson\IEEEauthorrefmark{2}}
\IEEEauthorblockA{\IEEEauthorrefmark{1} School of Electrical Engineering, Shanghai Jiao Tong University, Shanghai, China.}\IEEEauthorblockA{\IEEEauthorrefmark{2} Dyson School of Design Engineering, Imperial College London, London, United Kingdom.}
\vspace{-2.5 em}
}

\maketitle

\begin{abstract} 
Wind power producers can benefit from forming coalitions to participate cooperatively in electricity markets. To support such collaboration, various profit allocation rules rooted in cooperative game theory have been proposed. However, existing approaches overlook the lack of coherence among producers regarding forecast information, which may lead to ambiguity in offering and allocations. In this paper, we introduce a ``reconcile-then-optimize'' framework for cooperative market offerings. This framework first aligns the individual forecasts into a coherent joint forecast before determining market offers. With such forecasts, we formulate and solve a two-stage stochastic programming problem to derive both the aggregate offer and the corresponding scenario-based dual values for each trading hour. Based on these dual values, we construct a profit allocation rule that is budget-balanced and stable. Finally, we validate the proposed method through empirical case studies, demonstrating its practical effectiveness and theoretical soundness.
\end{abstract}



\section{Introduction}
Wind power producers (WPPs) typically participate in two trading floors within electricity markets and face challenges arising from the inherent uncertainty of wind generation \cite{morales2013integrating}. Aggregating wind energy resources across geographically dispersed locations can reduce power variability, which supports the view that WPPs can benefit from cooperative market participation \cite{bitar2011selling}. In this context, several studies have proposed market participation through aggregation managed by an external coordinator \cite{guerrero2015optimal}, followed by profit allocation among the participating WPPs. A variety of allocation rules have been investigated, including proportional sharing, the Shapley value, and the Nucleolus \cite{freire2014hybrid, nguyen2018sharing}.


However, most existing studies primarily focus on how to share profits once a joint offer is made, presuming a known and agreed-upon characterization of uncertainty. In practice, the distribution of wind power generation is unknown, and both WPPs and the external coordinator typically hold heterogeneous probabilistic forecasts—built from distinct models, data sources, and post-processing pipelines. Therefore, coherence (or consistency) of the pooled forecasts from different agents is a prerequisite for credible cooperation. Lacking coherence leads to conflicts between aggregate and individual scenarios, ambiguity in day-ahead offers, and fairness claims in ex-post allocations without a common reference point. Our premise is simple: no coherence, no credible cooperation.

The need for coherence is amplified by the well-documented properties of ensemble numerical weather prediction (NWP) systems used to generate wind power forecasts. Such systems often exhibit bias (systematic over- or under-prediction on average) and under-dispersion (a tendency to underestimate uncertainty) \cite{leutbecher2008ensemble, doubleday2020probabilistic}, which, in turn, undermines calibration and can degrade the profit of joint offering if left uncorrected. Therefore, prior to determining allocations, participants should first establish a common, coherent—and ideally better-calibrated—probabilistic representation of uncertainty by reconciling heterogeneous forecasts into a joint distribution.

Formally, coherence of hierarchical probabilistic forecasts requires that the forecasts lie within the appropriate linear subspace, ensuring consistency between aggregate and leaf levels. This implies that every probabilistic scenario satisfies the underlying structural constraints \cite{panagiotelis2023probabilistic}. To ensure coherence, a post-processing step known as \emph{forecast reconciliation} is performed, in which each scenario is adjusted to satisfy the structural constraints. For example, \cite{taieb2017coherent} proposed reordering scenarios using an empirical copula and subsequently aggregating them through the hierarchy. Other studies \cite{jeon2019probabilistic,panagiotelis2023probabilistic,tsiourvas2024learning} explored both fixed projection matrices and learning-based projection approaches for reconciliation. However, these projection-based approaches rely on the assumption that the base forecasts are unbiased. Given the under-dispersion nature of NWP systems, the effectiveness of purely projection-based reconciliation methods is limited.

\begin{figure*}[ht]
    \centering
    \includegraphics[width=0.75\linewidth]{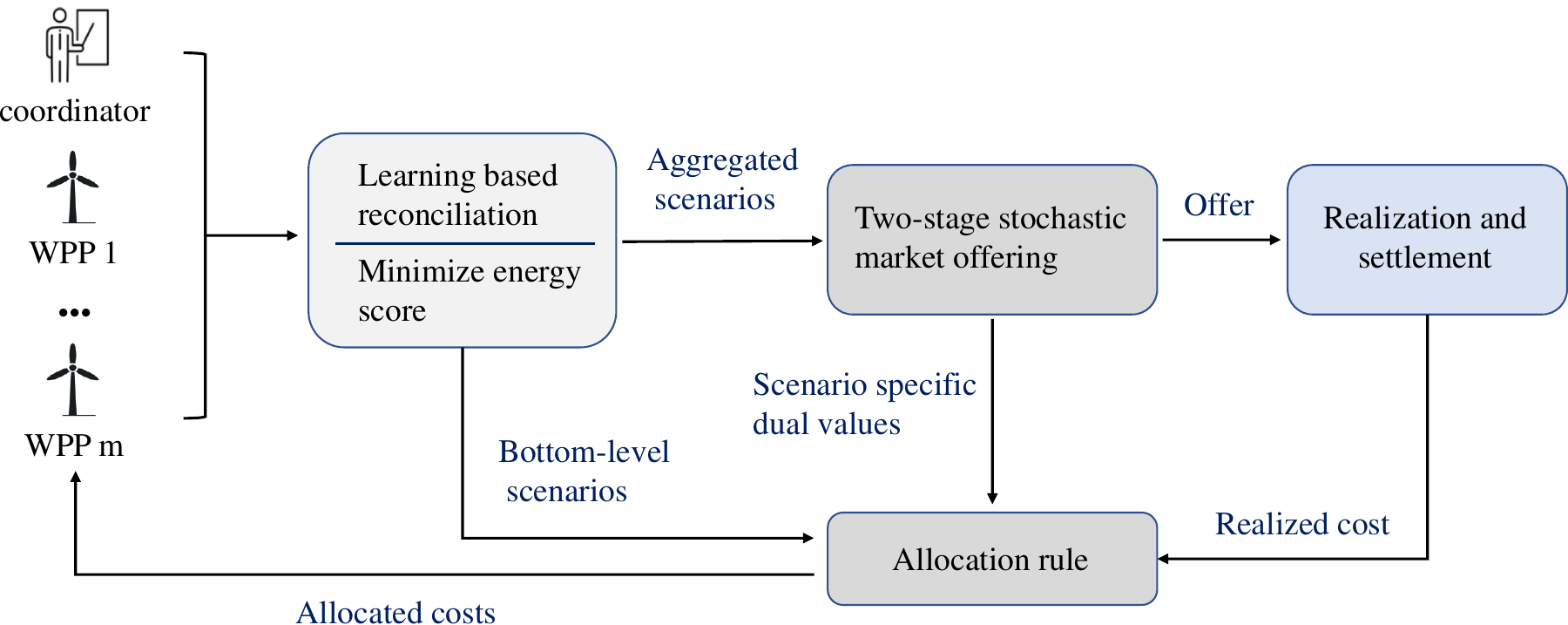}
    \caption{The proposed ``reconcile-then-optimize'' framework.}
    \label{fig:framework}
    \vspace{-1em}
\end{figure*}

In this work, we propose a ``reconcile-then-optimize'' paradigm, as illustrated in Fig.~\ref{fig:framework}: first align heterogeneous probabilistic forecasts into a coherent joint view; then optimize the market offer and allocate profits using that shared coherent information. Concretely, we develop a nonparametric, learning-based reconciliation model that maps each scenario to a coherent one. 
Unlike linear, projection-based reconciliation \cite{jeon2019probabilistic,panagiotelis2023probabilistic,tsiourvas2024learning}—which preserves distributional families and relies on unbiased bases—our nonparametric reconciler is a universal approximator. It can capture nonlinear cross-site dependencies and correct systematic bias/under-dispersion common in NWP-based forecasts. 
Using the reconciled scenarios, we solve a two-stage stochastic program for each trading hour to obtain the aggregate offer and the scenario-wise dual variables. A dual allocation mechanism, inspired by \cite{chen2009stochastic}, then produces an allocation that lies in the core (i.e., no subset of players could earn more by breaking away) and is budget-balanced, rendering cooperation both stable and implementable. In effect, the reconciled probabilistic forecasts jointly determine the market offer and the profit allocations. This procedure is computationally tractable and practically implementable, in contrast to Nucleolus-based methods. After the realization of actual generation, profits are distributed according to the shares specified by these allocations.

On a publicly available dataset (NYISO, day-ahead forecasts/observations) \cite{bryce2024solar}, the proposed method delivers better probabilistic forecast quality (e.g., lower energy score values and improved multivariate rank diagnostics) and, when embedded in the joint optimization, leads to higher realized profits across WPPs—highlighting that better forecasts translate into better decisions and profits.
The main contributions of this work are as follows: 
\begin{enumerate}
    \item We design a ``reconcile-then-optimize'' framework, placing emphasis on achieving coherence over shared probabilistic information prior to decision-making and profit allocation. By enforcing forecast coherence first, we eliminate informational disputes and can then ensure that the dual-based allocation from the joint stochastic program lies in the core and is budget-balanced, making the cooperation both credible and implementable. 
    \item We develop a nonparametric reconciliation method for probabilistic forecasts that possesses universal approximation capability, allowing it to capture nonlinear dependencies, correct bias, and improve calibration compared to linear projection methods.
\end{enumerate}
The remainder of this paper is organized as follows. Section \Romannum{2} introduces the proposed nonparametric probabilistic forecast reconciliation method and establishes its properties. Section \Romannum{3} presents the trading strategy and allocation mechanism. Section \Romannum{4} details the case study setup, including data descriptions, benchmarks, and evaluation metrics. Section \Romannum{5} reports the numerical results, and Section \Romannum{6} concludes the paper.

\section{Probabilistic Forecast Reconciliation}
Wind power producers (WPPs) can improve their market participation by cooperating in collective wind power offering \cite{bitar2011selling}. Much of the existing work, however, assumes that the distribution of wind generation is known. In practice, only probabilistic forecasts from heterogeneous sources are available. Consequently, reconciling these forecasts into a coherent representation is a prerequisite for effective collective decision-making. Moreover, because NWP-based probabilistic forecasts are often poorly calibrated \cite{bryce2024solar}, reconciliation is expected not only to enforce coherence but also to improve forecast quality \cite{zhang2024asset}.

In this study, the power generation of the WPP $i$ at time $t$ is modeled as a random variable $Y_{i,t}$, with realization $y_{i,t}$. We consider a set of WPPs $\mathcal{M}=\{1,2,\cdots,m\}$. The aggregate generation is represented by the random variable $Y_{\mathrm{sum},t} = \sum_{i=1}^m Y_{i,t}$, with its realization given by $y_{\mathrm{sum},t} = \sum_{i=1}^m y_{i,t}$.
For convenience, we denote $[y_{\mathrm{sum},t},y_{1,t},y_{2,t},\cdots,y_{m,t}]^\top$ compactly as $\mathbf{y}_t$, and $[y_{1,t},y_{2,t},\cdots,y_{m,t}]^\top$ as $\mathbf{b}_t$. Let the structure matrix $\mathbf{G}\in \mathbb{R}^{(m+1)\times m}$ be
\begin{equation*}
    \mathbf{G}=\begin{bmatrix}
        \mathbf{1}_m^\top\\
        \mathbf{I}_{m\times m}
    \end{bmatrix},
\end{equation*}
where $\mathbf{1}_m$ denotes an $m$-dimensional all-ones vector, and $\mathbf{I}_{m\times m}$ denotes the $m\times m$ identity matrix. Then, the relationship between the whole and the bottom-level time series can be expressed as
\begin{equation}
    \mathbf{y}_t = \mathbf{G}\mathbf{b}_t.
\end{equation}

Suppose we have probabilistic forecasts $\{\Tilde{y}_{i,t+k|t}^{(\xi)} \}_{\xi=1}^N$ for each series $i$ with lead time $k$, issued by heterogeneous agents. Particularly, in the context of day-ahead wind power forecasts, $t$ denotes the daily forecast issuance time, while the lead time $k$ spans the subsequent 24-hour horizon.\footnote{Typically, there is a gap between the forecast issuance time and midnight. For notational simplicity, however, we assume $t$ corresponds to midnight and let $k$ span from 1 to 24.} Accordingly, we express each scenario in compact form as
\begin{align*}
\Tilde{\mathbf{y}}_{t+k|t}^{(\xi)}&=\left[\Tilde{y}_{\mathrm{sum},t+k|t}^{(\xi)},\Tilde{y}_{1,t+k|t}^{(\xi)},\Tilde{y}_{2,t+k|t}^{(\xi)},\cdots,\Tilde{y}_{m,t+k|t}^{(\xi)} \right]^\top,\\
\Tilde{\mathbf{b}}_{t+k|t}^{(\xi)}&=\left[\Tilde{y}_{1,t+k|t}^{(\xi)},\Tilde{y}_{2,t+k|t}^{(\xi)},\cdots,\Tilde{y}_{m,t+k|t}^{(\xi)} \right]^\top.
\end{align*}
Then, we formally define the coherence of probabilistic forecasts.
\begin{definition}
[Coherence] A group of probabilistic forecasts $\{\Tilde{\mathbf{y}}_{t+k|t}^{(\xi)} \}_{\xi=1}^M$ is said to be coherent if $\Tilde{\mathbf{y}}_{t+k|t}^{(\xi)} = \mathbf{G}\Tilde{\mathbf{b}}_{t+k|t}^{(\xi)}, \forall \xi\in [N]$, where $[N]=\{ 1,2,\cdots,N\}$.
\end{definition}
\noindent This condition implies that the coherent probabilistic forecasts reside in the linear subspace of $\mathbb{R}^{m+1}$ spanned by the columns of $\mathbf{G}$.

\subsection{Learning-based Reconciliation Model}
Since the forecasts $\{\Tilde{y}_{i,t+k|t}^{(\xi)} \}_{\xi=1}^N$ are issued by heterogeneous agents, $\{\Tilde{\mathbf{y}}_{t+k|t}^{(\xi)} \}_{\xi=1}^N$ they may not be coherent by construction. Therefore, it is necessary to reconcile the forecasts to ensure coherence.
Indeed, achieving coherence between aggregate and individual forecasts is relatively straightforward due to their inherent two-level hierarchical structure. Basic reconciliation methods, such as the bottom-up approach— which aggregates individual forecasts scenario by scenario— can ensure coherence, but often at the expense of forecast quality. Inspired by \cite{panagiotelis2023probabilistic}, we propose a learning-based reconciliation approach. Specifically, letting $\hat{\mathbf{y}}_{t+k|t}^{(\xi)}$ denote the reconciled forecast corresponding to the base forecast $\Tilde{\mathbf{y}}_{t+k|t}^{(\xi)}$, we establish the reconciliation model $T_\mathrm{R}$ that operates on a scenario-by-scenario basis. It is expressed as
\begin{equation}
\hat{\mathbf{y}}_{t+k|t}^{(\xi)}=T_\mathrm{R}(\Tilde{\mathbf{y}}_{t+k|t}^{(\xi)};\theta),
\end{equation}
where $\theta$ denotes the parameters. More concretely, we model $T_\mathrm{R}$ as a combination of matrix multiplication and a learnable function $q$, defined as
\begin{align*}
    \mathbf{h}_{t+k}^{(\xi)}&=q(\Tilde{\mathbf{y}}_{t+k|t}^{(\xi)};\theta),\\
    \hat{\mathbf{y}}_{t+k|t}^{(\xi)}&= \mathbf{G}\mathbf{h}_{t+k}^{(\xi)},
\end{align*}
where $\mathbf{h}_{t+k}^{(\xi)}\in \mathbb{R}^m$ denotes an intermediate representation, which can be interpreted as the reconciled bottom-level forecast.

Existing works \cite{panagiotelis2023probabilistic,tsiourvas2024learning} model $q$ as a special linear function parameterized by a learnable projection matrix $\mathbf{Q}\in \mathbb{R}^{m\times (m+1)}$, i.e.,
\begin{equation*}
    \hat{\mathbf{y}}_{t+k|t}^{(\xi)}= \mathbf{G}\mathbf{Q}\Tilde{\mathbf{y}}_{t+k|t}^{(\xi)}.
\end{equation*}
In contrast, we model $q$ as a neural network, and establish that the resulting reconciliation operator $T_\mathrm{R}$ has the property of a universal approximator. To start with, we introduce the noise-outsourcing lemma in the view of normalizing flow \cite{papamakarios2021normalizing}.
\begin{lemma}
[Normalizing flow] Let $(X,Y)$ be a random pair taking values  in $\mathcal{X}\times \mathcal{Y}$ with joint distribution $F_{X,Y}$, where $\mathcal{Y}$ is assumed to be a standard Borel space. Then, there exits a random vector $\eta\sim \mathcal{N}(\mathbf{0},\mathbf{I}_d)$ where $d$ is the dimension of $Y$, and a Borel measurable and invertible function $T:\mathbb{R}^d\times \mathcal{X}\rightarrow \mathcal{Y}$ such that $\eta$ is independent of $X$ and
\begin{equation*}
    (X,Y)=(X,T(\eta,X))
\end{equation*}
almost surely.
\end{lemma}
\begin{proof}
This results directly from the general probability transformation formula, as described in Section 5 of \cite{papamakarios2021normalizing}.
\end{proof}
\noindent Then, each set of probabilistic forecasts $\{\Tilde{y}_{i,t+k|t}^{(\xi)} \}_{\xi=1}^M$ is induced from a generative model of the form
\begin{equation*}
    T_i(\eta_i,X_{i,t}),
\end{equation*}
where $\eta_i \sim \mathcal{N}(0,1)$ and $X_{i,t}$ denotes the contextual information (in this case, NWP features). Let $\mathbf{X}_t$ and $\mathbf{Y}_t$ denote the concatenations of $\{X_{i,t}\}_{i=1}^m$ and $\{Y_{i,t}\}_{i=1}^m$  respectively. We then present our results as follows.
\begin{proposition}
[Universal approximator] The nonparametric transform $T_\mathrm{R}$ is a universal approximator of the joint distribution $F_{\mathbf{Y}_{t+k}|\mathbf{X}_t}$.
\end{proposition}
\begin{proof}
The scenarios $\{T_\mathrm{R}(\Tilde{\mathbf{y}}_{t+k|t}^{(\xi)}) \}_{\xi=1}^N$ are induced from the model
\begin{equation*}
    T_\mathrm{R}(T_1(\eta_1,X_{1,t}),T_2(\eta_2,X_{2,t})\cdots,T_m(\eta_m,X_{m,t})),
\end{equation*}
which is equivalent to 
\begin{equation*}
    T(\eta,\mathbf{X}_t),
\end{equation*}
where $\eta \sim \mathcal{N}(\mathbf{0},\mathbf{I}_m)$ and $T$ is an invertible function, as justified by the universal approximation theorem for neural networks. By the noise-outsourcing lemma, $(\mathbf{X}_t,T(\eta,\mathbf{X}_t))$ is equal to $(\mathbf{X}_t,\mathbf{Y}_{t+k})$ almost surely. Consequently, $T_\mathrm{R}$ serves as a universal approximator of the conditional distribution $F_{\mathbf{Y}_{t+k}|\mathbf{X}_t}$.
\end{proof}

Furthermore, the reconciliation model $T_\mathrm{R}$ aggregates information across the hierarchy, thereby enhancing forecast quality. In contrast to our reconciliation method, the projection-based approach \cite{panagiotelis2023probabilistic,tsiourvas2024learning} is not a universal approximator, as it restricts the forecasts to remain within the same family of distributions. Especially when the marginal distributions are misspecified, the projection reconciliation approach must lead to biased joint distribution estimates. 

\subsection{Parameter Estimation}
To estimate the parameters $\theta$ of the reconciliation model $T_\mathrm{R}$, we adopt \emph{distribution matching}. Because only realizations $\mathbf{y}_t$ rather than the full distribution $F_{\mathbf{Y}_t}$, are observed, we implement distribution matching using the energy score.
Given a predictive distribution represented by scenarios $\{\hat{\mathbf{y}}_{t+k|t}^{(\xi)} \}_{\xi=1}^N$, the energy score can be expressed in scenario form. Specifically, we generate two independent random permutations of the indices $1,2,\cdots,N$, denoted as $r_1,r_2,\cdots,r_N$ and $s_1,s_2,\cdots,s_N$. The energy score is then defined as
\begin{equation}
\begin{aligned}
    \operatorname{ES}&\left(\hat{F}_{\mathbf{Y}_{t+k|t}},\mathbf{y}_{t+k}\right)=\frac{1}{N}\sum_{i=1}^N||\hat{\mathbf{y}}_{t+k|t}^{(r_i)}-\mathbf{y}_{t+k}||_2 \\
    &-\frac{1}{2N^2}\sum_{i=1}^N\sum_{j=1}^N||\hat{\mathbf{y}}_{t+k|t}^{(r_i)}-\hat{\mathbf{y}}_{t+k|t}^{(s_j)}||_2.
\end{aligned}
\end{equation}
where $||\cdot||_2$ is the Euclidean norm.
As shown in \cite{gneiting2007strictly}, the energy score is a strictly proper scoring rule, meaning that its expected value is uniquely minimized when the predictive distribution coincides with the true distribution. 
Suppose we have access to a training set
\begin{equation*}
\Big\{ \big( \{\tilde{\mathbf{y}}_{t+k|t}^{(\xi)} \}_{\xi=1}^N, \mathbf{y}_{t+k} \big)\big| t \in \mathcal{T}_\mathrm{tr}, k = 1,2,\ldots,24 \Big\},
\end{equation*}
where $\mathcal{
T
}_\mathrm{tr} $ denotes the set of forecast issuance times used for training. Particularly, we employ the same reconciliation model $T_\mathrm{R}$ across all lead time $k$. The parameters $\theta$ are then estimated by minimizing the average energy score, 
\begin{equation}
    \theta^*=\arg \min_\theta \ \frac{1}{24\cdot|\mathcal{T}_\mathrm{tr}|}\sum_{k=1}^{24}\sum_{t\in \mathcal{T}_\mathrm{tr}}\operatorname{ES}\left(\hat{F}_{\mathbf{Y}_{t+k|t}},\mathbf{y}_{t+k}\right).
\end{equation}

\section{Trading Strategy and Allocation Mechanism}

In particular, we consider two trading floors: the day-ahead (forward) market and the real-time (balancing) market, the latter operating under a dual-price settlement mechanism. Throughout, we assume that WPPs act as price takers and are risk-neutral. At the day-ahead stage, each power producer is required to submit market offers at the closure time $t$, covering the 24 hours of the following day. For example, the offer of WPP $i$ submitted for trading hour $t+k$ is denoted by $y_{i,t+k}^\mathrm{F}$. Since power systems must balance supply and demand in real time, deviations between the offer $y_{i,t+k}^\mathrm{F}$ and actual generation $y_{i,t+k}$ are settled in the balancing stage before delivery. 

\subsection{Trading Strategy of Single WPP}

Under the dual-price settlement mechanism, let $\pi_{t+k}^\mathrm{UP}$ and $\pi_{t+k}^\mathrm{DW}$ denote the upward and downward regulation prices for trading hours $t+k$, and let $\pi_{t+k}^\mathrm{F}$ denote the day-ahead price. By the design of the mechanism, the relation $\pi_{t+k}^\mathrm{UP}\geq\pi_{t+k}^\mathrm{F}\geq\pi_{t+k}^\mathrm{DW}$ holds. 
Following common practice in the renewable energy offering literature \cite{morales2013integrating},  the profit of the WPP $i$ at trading hour $t+k$ can be written in terms of deviation penalties as
\begin{equation}
\begin{aligned}
\rho_{i,t+k}^{(\mathrm{I})}=\pi_{t+k}^\mathrm{F} y_{i,t+k} &- 
\biggl[ \psi_{t+k}^-  [y_{i,t+k}^\mathrm{F} - y_{i,t+k}]^+ \\
&+ \psi_{t+k}^+ [y_{i,t+k} - y_{i,t+k}^\mathrm{F}]^+ \biggr],
\end{aligned}
\label{independent_profit}
\end{equation}
where $[x]^+=\max(x,0)$, $\psi_{t+k}^+=\pi_{t+k}^\mathrm{F}-\pi_{t+k}^\mathrm{DW}$, and $\psi_{t+k}^-=\pi_{t+k}^\mathrm{UP}-\pi_{t+k}^\mathrm{F}$. For convenience, we define the second term of the RHS as the balancing cost, denoted by $c_{t+k}(y_{i,t+k}^\mathrm{F}, y_{i,t+k})$. Since the first term of (\ref{independent_profit}) is beyond the producer's control, the optimal offer is obtained by minimizing the expected balancing cost.

If the distribution of $Y_{i,t+k}$ is known, the optimal offer is obtained as
\begin{equation}
y_{i,t+k}^{\texttt{F}*}=\underset{y_{i,t+k}^\texttt{F}\in[0,P_{i,\mathrm{rated}}]}{\arg\min}\ \mathbb{E}_{Y_{i,t+k}}\left[ c_{t+k}(y_{i,t+k}^\texttt{F},y_{i,t+k})\right],
\label{exp_loss_min}
\end{equation}
where $P_{i,\mathrm{rated}}$ denotes the rated capacity of the WPP $i$. In particular, \cite{pinson2007trading} demonstrated that problem (\ref{exp_loss_min}) can be cast as a newsvendor formulation, yielding an analytical solution expressed in terms of quantiles:
\begin{equation}
y_{i,t+k}^{\texttt{F}*}=F^{-1}_{y_{i,t+k}}\left(\frac{\psi_{t+k}^+}{\psi_{t+k}^++\psi_{t+k}^-}\right),
\label{quantile}
\end{equation}
where $F^{-1}(\cdot)$ represents the inverse CDF. On the other hand, when a set of $N$ probabilistic scenarios $\{ y_{i,t+k}^{(\xi)} \}_{\xi=1}^N$ is available, the offer can alternatively be determined through a stochastic programming approach \cite{morales2010short}.

\subsection{Cooperative Wind Power Offering}
If the wind power generation assets are owned by a single stakeholder, the optimal aggregate offer can be directly determined as the quantile of $Y_{\mathrm{sum},t}$, analogous to (\ref{quantile}). However, when the assets belong to different stakeholders, it becomes necessary to allocate the resulting profits among them.
With the reconciled forecasts $\{\hat{\mathbf{y}}_{t+k|t}^{(\xi)} \}_{\xi=1}^N$ at hand, we now describe how the collective of WPPs participates in the electricity markets and allocates the resulting costs. For the grand coalition $\mathcal{M}$, let $y_{\mathrm{sum},t+k}^\mathrm{F}$ denote the aggregate offering quantity. Then the corresponding two-stage stochastic programming problem can be formulated as follows:
\begin{equation}
\begin{aligned}
\min_{y_{\mathrm{sum},t+k}^\mathrm{F}} \quad & \frac{1}{N}\sum_{\xi=1}^{N} (\psi_{t+k}^+ u^{(\xi)}_{\mathrm{sum},t+k}+\psi_{t+k}^- w^{(\xi)}_{\mathrm{sum},t+k}) \\
\mathrm{s.t.} \quad 
& y_{\mathrm{sum},t+k}^\mathrm{F} + u^{(\xi)}_{\mathrm{sum},t+k}-w^{(\xi)}_{\mathrm{sum},t+k} \\
& \quad=\hat{y}_{\mathrm{sum},t+k|t}^{(\xi)}, \quad \forall \xi \in [N] \quad &(\nu_{t+k}^{(\xi)})\\
& u^{(\xi)}_{\mathrm{sum},t+k} \geq 0, \, w^{(\xi)}_{\mathrm{sum},t+k} \geq 0, \quad \forall \xi \in [N] \\
& 0 \leq y_{\mathrm{sum},t+k}^\mathrm{F} \leq P_{\mathrm{sum},\mathrm{rated}},
\end{aligned}
\label{P}
\end{equation}
where $u^{(\xi)}_{\mathrm{sum},t+k}$ and $w^{(\xi)}_{\mathrm{sum},t+k}$ denote the aggregated over- and under-production quantities in scenario $\xi$, $P_{\mathrm{sum},\mathrm{rated}}=\sum_{i\in \mathcal{M}}P_{i,\mathrm{rated}}$ denotes the total rated capacity, and $\nu^{(\xi)}_{t+k}$ represents the corresponding dual variable. We can further define a value function  for the expected cost $v:\mathbb{R}^M\rightarrow \mathbb{R}$, expressed as 
\begin{equation*}
l_{t+k}(\mathcal{M})=v(\hat{y}_{\mathrm{sum},t+k|t}^{(1)},\hat{y}_{\mathrm{sum},t+k|t}^{(2)},\cdots,\hat{y}_{\mathrm{sum},t+k|t}^{(N)}).
\end{equation*} Accordingly, the dual problem of (\ref{P}) can be derived as
\begin{equation}
\begin{aligned}
\max_{\{\nu_{t+k}^{(\xi)}\}} \quad & \frac{1}{N}\sum_{\xi =1}^N \nu^{(\xi)}_{t+k}\hat{y}_{\mathrm{sum},t+k}^{(\xi)} \\
\mathrm{s.t.} \quad 
& \frac{1}{N}\sum_{\xi =1}^N\nu_{t+k}^{(\xi)}\leq 0,\\
& -\psi^-_{t+k} \leq \nu_{t+k}^{(\xi)} \leq \psi_{t+k}^+, \quad \forall \xi \in [N]
\end{aligned}
\label{D}
\end{equation}

Both problems (\ref{P}) and (\ref{D}) are linear and can therefore be solved efficiently. Solving them yields $y_{\mathrm{sum},t+k}^{\mathrm{F}*}$, the aggregate offer submitted to the market, together with $\{\nu_{t+k}^{(\xi)*} \}_{\xi=1}^N$, the dual variables associated with each scenario. We then define the allocated expected cost $a_{i,t+k}$ for WPP $i$ at trading hour $t+k$, as
\begin{equation}
    a_{i,t+k}=\frac{1}{N}\sum_{\xi=1}^N \hat{y}_{i,t+k}^{(\xi)}\nu_{t+k}^{(\xi)*}.
\end{equation}
Building upon the result of \cite{chen2009stochastic}, we now present an allocation that lies within the core.
\begin{proposition}
[Core allocation] Given the reconciled forecasts $\{ \hat{\mathbf{y}}_{t+k|t}^{(\xi)} \}_{\xi=1}^N$, the vector $\mathbf{a}_{t+k}=\left[a_{1,t+k},\cdots,a_{m,t+k} \right]^\top$ constitutes an allocation in the core of the wind power offering game $(\mathcal{M},l)$.
\end{proposition}
\begin{proof}
By the strong duality for the grand coalition $\mathcal{M}$, we have
\begin{equation}
\begin{aligned}
    l_{t+k}(\mathcal{M})&=\frac{1}{N}\sum_{\xi =1}^N \nu^{(\xi)*}_{t+k}\hat{y}_{\mathrm{sum},t+k}^{(\xi)}\\
    &=\sum_{i\in \mathcal{M}}\left(\frac{1}{N}\sum_{\xi =1}^N \nu^{(\xi)*}_{t+k}\hat{y}_{i,t+k}^{(\xi)}\right)=\sum_{i\in \mathcal{M}}a_{i,t+k},
\end{aligned}
\label{efficiency}
\end{equation}
which means the allocation is \emph{efficient}. For any sub-coalition $\mathcal{S}$, its decision-making is based on the information set $\left \{ \{\hat{y}_{i,t+k|t}^{(\xi)} \}_{\xi=1}^N | i\in \mathcal{S} \right\}$. In this context, we define the aggregated probabilistic scenarios for a coalition $\mathcal{S}$ as $\sum_{i \in \mathcal{S}}\hat{y}_{i,t+k}^{(\xi)}$, and the aggregate offering quantity as $y_{\mathcal{S},t+k}^{\mathrm{F}*}$. Accordingly, the characteristic function of the wind power offering game is given by $l_{t+k}(\mathcal{S})$. By forming a counterpart problem of (\ref{D}) for the coalition $\mathcal{S}$, the constraints are of the same type, written only for $\mathcal{S}$. Thus, $\{\nu^{(\xi)*}_{t+k}\}_{\xi=1}^N$ is still feasible for this counterpart problem. By the weak duality for the coalition $\mathcal{S}$, we have
\begin{align}
    \sum_{i\in \mathcal{S}} a_{i,t+k}=\sum_{i\in \mathcal{S}}\left(\frac{1}{N}\sum_{\xi =1}^N \nu^{(\xi)*}_{t+k}\hat{y}_{i,t+k}^{(\xi)}\right)\leq l_{t+k}(\mathcal{S}).
    \label{rationality}
\end{align}
It means that every coalition $\mathcal{S}$ receives at most its stand-alone value. The conditions (\ref{efficiency}) and (\ref{rationality}) collectively satisfy the definition of the core.
\end{proof}


We note that the above result relies on the coherence of the forecasts $\{ \hat{\mathbf{y}}_{t+k|t}^{(\xi)} \}_{\xi=1}^N$. However, the allocation derived above reflects only the expected shared cost for each WPP. The realized imbalance cost becomes known only after actual generation is observed. This distinction motivates the development of ex-post profit allocation mechanisms that account for realized outcomes. Given the aggregate offer $y^{\mathrm{F}*}_{\mathrm{sum},t+k}$ and realized generation $y_{\mathrm{sum},t+k}$, the real aggregate imbalance cost is $c_{t+k}(y^{\mathrm{F}*}_{\mathrm{sum},t+k},y_{\mathrm{sum},t+k})$, which quantifies the penalty due to deviations between the aggregate offer and actual generation. We then define the realized cost allocation $c_{i,t+k}$ for each WPP $i$, derived from the allocation vector $\mathbf{a}_{t+k}$ as
\begin{equation}
    c_{i,t+k}=\frac{a_{i,t+k}}{\sum_{i=1}^ma_{i,t+k}} \, c_{t+k}(y^{\mathrm{F}*}_{\mathrm{sum},t+k},y_{\mathrm{sum},t+k}).
\end{equation}
Accordingly, the profit $\rho_{i,t+k}^{(\mathrm{C})}$ for the WPP $i$ under the ``reconcile-then-optimize'' framework is derived as
\begin{equation}
\rho_{i,t+k}^{(\mathrm{C})}=\pi_{t+k}^\mathrm{F}y_{i,t+k}-c_{i,t+k}.
\label{cooperative_profit}
\end{equation}

\begin{proposition}
[Ex-post allocation] The ex-post allocation is budget-balanced and stable with respect to the expected cost.    
\end{proposition}
\begin{proof}
It is apparent that the total realized cost for every WPP matches the realized cost derived from the market, i.e.,
\begin{equation}
    \sum_{i\in \mathcal{M}}c_{i,t+k}=c_{t+k}(y^{\mathrm{F}*}_{\mathrm{sum},t+k},y_{\mathrm{sum},t+k}).
\end{equation}
In addition, we have $\mathbb{E}\left[c_{i,t+k} \right]=a_{i,t+k}$, since 
\begin{equation*}
    \mathbb{E}\left[ c_{t+k}\left(y^{\mathrm{F}*}_{\mathrm{sum},t+k},y_{\mathrm{sum},t+k}\right)\right]=l_{t+k}(\mathcal{M}).
\end{equation*}
It then follows that
\begin{equation}
    \sum_{i\in \mathcal{S}}\mathbb{E}\left[c_{i,t+k} \right]\leq \mathbb{E}\left[ c_{t+k}\left(y_{\mathcal{S},t+k}^{\mathrm{F}*},\sum_{i\in \mathcal{S}}y_{i,t+k}\right)\right],
\end{equation}
which means the allocation is stable from the perspective of expected cost. 
\end{proof}

\begin{figure*}[ht!]
    \centering
    \includegraphics[width=0.78\linewidth]{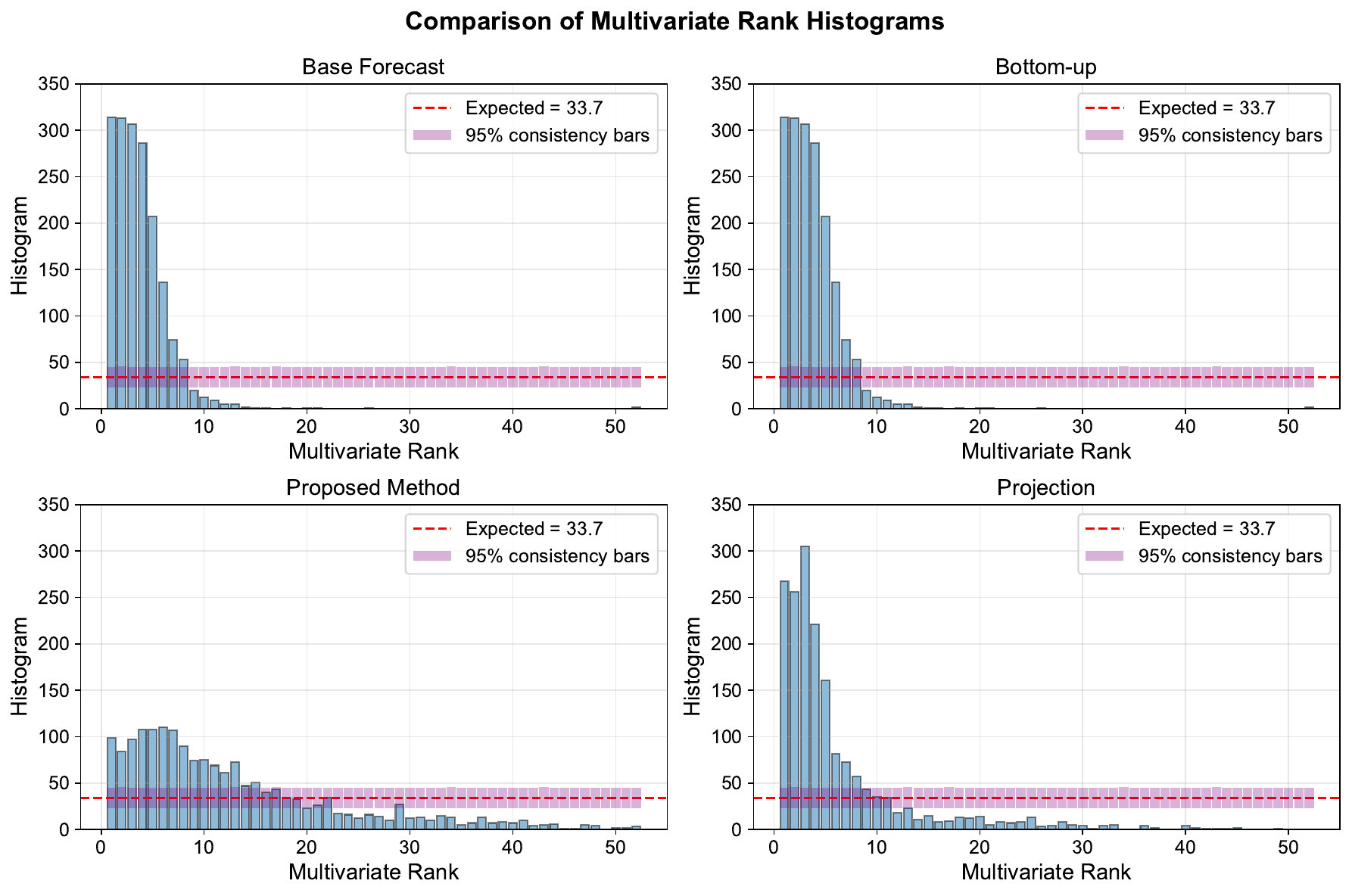}
    \caption{Multivariate rank histograms for different methods (consistency bars were obtained through simulation of perfectly calibrated forecasts \cite{brocker2007increasing}).}
    \label{fig:Histograms}
    \vspace{-1em}
\end{figure*}

\section{Case study}

\subsection{Setups and Benchmarks}

We utilize the open dataset released by the National Renewable Energy Laboratory (NREL) \cite{bryce2024solar}. The site-level day-ahead wind power forecasts are generated using weather forecast data from the European Centre for Medium-Range Weather Forecasts (ECMWF), consisting of 51 ensemble members. Zonal-level forecasts are subsequently obtained through aggregation and Bayesian model averaging. The `actual' wind power generation data are simulated using a physical model and an associated database.
Specifically, we select the day-ahead wind power forecasts and corresponding observations for the North zone within the NYISO balancing area for the year 2018. The dataset comprises 8 wind farms distributed across two geographic clusters, with their capacities summarized in Table~\ref{tab:capacity}. 
We use the first 80\% of the data for model training and reserve the remaining 20\% for out-of-sample validation.

\begin{table}[htbp]
\vspace{-1em}
\centering
\caption{Wind farm capacities (MW)}
\begin{tabular}{ll|ll}
\toprule
Wind Farm & Capacity & Wind Farm & Capacity \\
\midrule
Marble\_River & 215.25 & Noble\_Chateaugay & 106.5 \\
Noble\_Clinton & 100.5 & Jericho\_Rise & 77.7 \\
Noble\_Ellenburg & 81 & Bull\_Run\_II\_Wind & 145.4 \\
Noble\_Altona & 97.5 & Bull\_Run\_Wind & 303.6 \\
\midrule
\bottomrule
\end{tabular}
\label{tab:capacity}
\vspace{-1em}
\end{table}

In this work, we investigate a simplified setting in which the day-ahead price and balancing penalties are assumed to be fixed. This setup reflects certain specific instances observed in electricity markets. Particularly, we consider 
$\pi^\mathrm{F}=25$ \$/MWh, $\psi^+=4$ \$/MWh, $\psi^-=12$ \$/MWh. Evidently, the optimal offer corresponds to 
the 0.25 quantile in this case. We evaluate several benchmarks, including independent offering and projection-based counterparts, which are detailed as follows.

\subsubsection{Independent offering} Each WPP trades independently, using the corresponding quantiles derived from its own forecasts.
\subsubsection{Bottom-up reconciliation} All WPPs naively pool their forecasts and aggregate them through the hierarchy on a scenario-by-scenario basis. The resulting reconciled forecasts are then used within the cooperative wind power offering to determine the joint offer and corresponding allocations.
\subsubsection{Projection-based reconciliation \cite{jeon2019probabilistic,panagiotelis2023probabilistic,tsiourvas2024learning}} All WPPs and the coordinator pool their forecasts using a linear projection-based reconciliation model, with parameters optimized by minimizing the energy score. The reconciled forecasts are subsequently fed into the cooperative wind power offering framework.

\vspace{-1em}
\subsection{Evaluation Metrics}
We evaluate both the quality of reconciled forecasts and the resulting profits of each trading strategy on an out-of-sample dataset, $\Big\{ \big( \{\tilde{\mathbf{y}}_{t+k|t}^{(\xi)} \}_{\xi=1}^M, \mathbf{y}_{t+k} \big)\big| t \in \mathcal{T}_\mathrm{te}, k = 1,2,\ldots,24 \Big\}$,
where $\mathcal{T}_\mathrm{te} $ denotes the set of forecast issuance times used for training. In particular, we also evaluate the calibration of the reconciled forecasts, as this property is crucial for the quality of subsequent decisions. The evaluation metrics are detailed as follows.
\subsubsection{Energy Score}
We report the average energy score (AES) across the test set, i.e.,
\begin{equation*}
    \mathrm{AES}=\frac{1}{24 |\mathcal{T}_\mathrm{te}|}\sum_{k=1}^{24}\sum_{t\in \mathcal{T}_\mathrm{te}}\operatorname{ES}\left(\hat{F}_{\mathbf{Y}_{t+k|t}},\mathbf{y}_{t+k}\right).
\end{equation*}
\subsubsection{Calibration}
Calibration quantifies the statistical consistency between a probabilistic forecast and the corresponding realized observation. We employ the band-depth rank histogram proposed by \cite{thorarinsdottir2016assessing}. Specifically, we compute the band depth of the observation with respect to the set of scenarios, as well as the band depth of each ensemble member. The multivariate verification rank of the observation is then defined as its position within the ordered set of all computed band depths. Further methodological details can be found in \cite{thorarinsdottir2016assessing}.
\subsubsection{Realized Profits}
We report the average profit (AP) of each WPP across the test set, defined as
\begin{equation*}
\mathrm{AP}_i =
\frac{1}{24 |\mathcal{T}_\mathrm{te}|}
\sum_{k=1}^{24} \sum_{t \in \mathcal{T}_\mathrm{te}}
\rho_{i,t+k}^{(\cdot)},
\end{equation*}
where $\rho_{i,t+k}^{(\cdot)}$ is given by (\ref{independent_profit}) or (\ref{cooperative_profit}), depending on the trading strategy employed.

\section{Results}
\subsection{Quality and Calibration}

We present the energy scores for different models in Table~\ref{tab:energy_scores}. In this setup, since the coordinator’s base forecasts are of higher quality than those of the WPPs, the energy score of the base forecasts is even lower than that obtained through bottom-up reconciliation. Both the proposed and projection-based reconciliation models are learning-based, enabling the reconciled forecasts to outperform the base forecasts. Specifically, the proposed method attains the best performance, attributable to its universal approximation capability, which allows it to capture nonlinear dependencies and correct systematic biases.

\begin{table}[htbp]
\centering
\caption{Energy Scores across different methods (MW)}
\label{tab:energy_scores}
\begin{tabular}{lcccc}
\toprule
 & Base Forecast & Bottom-up & Projection & Proposed \\
\midrule
Energy score & 264.68 & 415.08 & 206.67 & \textbf{156.46} \\
\midrule
\bottomrule
\end{tabular}
\vspace{-1em}
\end{table}

In addition, we present the multivariate rank histogram as well as the 95\% consistency bars in Fig.~\ref{fig:Histograms}. Ideally, the ranks should follow a uniform distribution if the forecasts are well calibrated. In contrast, under-dispersive forecasts or forecasts with a systematic bias typically produce rank histograms with an excess of observations in the lower ranks \cite{thorarinsdottir2016assessing}. Therefore, both the base forecasts and the reconciled forecasts exhibit under-dispersion, likely due to the stochastic nature of the bias in the NWP system. Furthermore, we present the deviations between the multivariate rank frequencies and the uniform distribution in Table~\ref {tab:rank_scores}. Nevertheless, whereas bottom-up and projection-based reconciliation methods yield only marginal improvements in calibration, the proposed method achieves a substantial enhancement.

\begin{table}[htbp]
\vspace{-1em}
\centering
\caption{Deviations between the multivariate rank frequencies and the uniform distribution across different methods.}
\label{tab:rank_scores}
\begin{tabular}{lcccc}
\toprule
 & Base Forecast & Bottom-up & Projection & Proposed \\
\midrule
Deviation & 1.62 & 1.62 & 1.33 & \textbf{0.87} \\
\midrule
\bottomrule
\end{tabular}
\vspace{-2em}
\end{table}

\subsection{Realized Profits}

The average profits for each hour at each WPP are reported in Table \ref{tab:case 2}. It is observed that trading as an aggregation generally yields higher profits for nearly all WPPs, with the exception of the bottom-up reconciliation, where some WPPs achieve higher profits under the independent offering. However, it is worth noting that the expected imbalance costs for each WPP under bottom-up reconciliation remain lower than those incurred under the independent offering. Therefore, this distinction may be attributed to the fact that the realized costs have not yet converged to their expected values. Furthermore, both the proposed and projection-based models consistently yield higher profits for every WPP, with the proposed model achieving the best overall performance.

\begin{table}[htbp]
\centering
\caption{Comparison of average profits for each WPP (\$/hour)}
\label{tab:case 2}
\begin{tabular}{lcccc}
\toprule
Wind Site & Independent & Bottom-up & Projection & Proposed  \\
\midrule
Marble\_River & 2506.83 & 2537.76 & 2612.09 & \textbf{2627.69} \\
Noble\_Clinton & 1048.35 & 1057.95 & 1087.93 & \textbf{1095.54} \\
Noble\_Ellenburg & 847.50 & 833.23 & 866.13 & \textbf{875.34} \\
Noble\_Altona & 927.07 & 955.81 & 969.62 & \textbf{980.73} \\
Noble\_Chateaugay & 1070.72 & 1085.39 & 1105.91 & \textbf{1123.75} \\
Jericho\_Rise & 985.76 & 951.39 & 990.66 & \textbf{1007.17} \\
Bull\_Run\_II\_Wind & 1678.15 & 1864.22 & 1847.89 & \textbf{1861.10} \\
Bull\_Run\_Wind & 2559.20 & 2437.21 & 2591.69 & \textbf{2610.66} \\
\midrule
\bottomrule
\end{tabular}
\end{table}

\section{Conclusions}
In this work, we propose a ``reconcile-then-optimize'' framework for cooperative wind power offering, with a primary emphasis on achieving coherence over shared probabilistic information. Specifically, we develop a nonparametric reconciliation method with universal approximation capability, optimized using the (proper) energy score. Based on the reconciled forecasts, we construct a core allocation that is both theoretically grounded and practically implementable. The proposed framework is demonstrated on an open dataset from NREL. Results show that the method produces reconciled forecasts with superior forecast quality and calibration. Furthermore, in a simplified market simulation, the proposed approach consistently outperforms both independent offering strategies and existing reconciliation counterparts. 

However, the results also reveal that, although the proposed method improves calibration, the reconciled forecasts remain under-dispersive. This finding highlights the need for further research on enhancing calibration within forecast reconciliation frameworks. Moreover, the proposed reconciliation framework assumes full access to the observations of each WPP. Given the multi-agent nature of cooperative wind power offering, it would be valuable to investigate distributed reconciliation methods that respect data decentralization and privacy constraints.

\bibliographystyle{IEEEtran}
\bibliography{IEEEabrv,mylib}

\end{document}